\documentclass[conference, 10pt, twocolumn]{ieeeconf}       

\def\BibTeX{{\rm B\kern-.05em{\sc i\kern-.025em b}\kern-.08em
    T\kern-.1667em\lower.7ex\hbox{E}\kern-.125emX}}
    
\usepackage[left=54pt, right=54pt,bottom=54pt, top=54pt]{geometry}

\usepackage{amsmath,mathrsfs,amsfonts,amssymb,graphicx,epsfig}
 \usepackage{amsthm}
\usepackage{subcaption}
\usepackage{color,multirow,rotating}
\usepackage{algorithm,algpseudocode,algorithmicx}
\usepackage{cite,url,framed,bm,balance,stmaryrd}

\newtheorem{theorem}{Theorem}

\newtheorem{definition}{Definition}

\usepackage{tikz}
\usetikzlibrary{calc,trees,positioning,arrows,chains,shapes.geometric,decorations.pathreplacing,decorations.pathmorphing,shapes, matrix,shapes.symbols}

\makeatletter
\newcommand{\dotminus}{\mathbin{\text{\@dotminus}}}

\newcommand{\@dotminus}{%
  \ooalign{\hidewidth\raise1ex\hbox{.}\hidewidth\cr$\m@th-$\cr}%
}

\allowdisplaybreaks

\title{\LARGE\textbf{
Convex and Nonconvex Sublinear Regression with Application to Data-driven Learning of Reach Sets}
}

\author{Shadi Haddad and Abhishek Halder
\thanks{Shadi Haddad and Abhishek Halder are with the Department of Applied Mathematics, University of California, Santa Cruz, CA 95064, USA, {\tt\small{\{shhaddad,halder\}@ucsc.edu}}.%
}}

\IEEEoverridecommandlockouts
\begin{document}

\maketitle
\pagenumbering{arabic}

\begin{abstract}
We consider estimating a compact set from finite data by approximating the support function of that set via sublinear regression. Support functions uniquely characterize a compact set up to closure of convexification, and are sublinear (convex as well as positive homogeneous of degree one). Conversely, any sublinear function is the support function of a compact set. We leverage this property to transcribe the task of learning a compact set to that of learning its support function. We propose two algorithms to perform the sublinear regression, one via convex and another via nonconvex programming. The convex programming approach involves solving a quadratic program (QP). The nonconvex programming approach involves training a input sublinear neural network. We illustrate the proposed methods via numerical examples on learning the reach sets of controlled dynamics subject to set-valued input uncertainties from trajectory data.    
\end{abstract}


\section{Introduction}\label{sec:introduction}
Motivated by the correspondence between compact sets and their support functions, in this work, we consider computationally learning compact sets by performing regression on their support functions from finite data. Our main idea is to algorithmically leverage the isomorphism between the support functions and the space of sublinear functions -- a subclass of convex functions.

Several works in the optimization, learning and statistics literature \cite{prince1990reconstructing,prince1991convex,fisher1997estimation,cai2018adaptive,kur2020suboptimality,soh2021fitting} have investigated the problem of estimating compact sets up to convexification from experimentally measured or numerically simulated (possibly noisy) data. While the problem is of interest across a broad range of applications (e.g., obstacle detection from range measurements, non-intrusive fault detection in materials and manufacturing applications, tomographic imaging in medical applications), we are primarily motivated in data-driven learning of reach sets for safety-critical systems-control applications.

Formally, the (forward in time) reach set $\mathcal{X}_{t}$ is defined as the set of states a controlled dynamical system may reach at a given time $t>0$ subject to a controlled deterministic dynamics $\dot{\bm{x}}=\bm{f}(t,\bm{x},\bm{u})$ where the state vector $\bm{x}\in\mathbb{R}^{d}$, the feasible input $\bm{u}(t)\in\:\text{compact}\; \mathcal{U}\subset\mathbb{R}^{m}$, and the initial condition $\bm{x}(t=0)\in\:\text{compact}\:\mathcal{X}_{0}\subset\mathbb{R}^{d}$.  Specifically,
\begin{align}
\!\mathcal{X}_{t}:=\!\!\!\!\!&\bigcup_{\text{measurable}\;\bm{u}(\cdot)\in\mathcal{U}}\!\!\!\!\!\!\!\!\!\{\bm{x}(t)\in\mathbb{R}^{d} \mid \dot{\bm{x}}=\bm{f}(t,\bm{x},\bm{u}), \: \bm{x}(t=0)\in\mathcal{X}_0, \nonumber\\
&\qquad\qquad\bm{u}(\tau)\in\mathcal{U}\;\text{for all}\; 0\leq \tau\leq t\}.
\label{DefReachSets}    
\end{align}
With the aforesaid assumptions on $\mathcal{X}_0,\mathcal{U}$ in place, we suppose that the vector field $\bm{f}$ is sufficiently smooth to guarantee compactness of $\mathcal{X}_t$ for all $t\geq 0$.

In the systems-control literature, there exists a vast body of works (see e.g., \cite{althoff2021set,chutinan1999verification,villanueva2015unified,le2010reachability,pecsvaradi1971reachable} as representative references) on reach set computation. Interests behind approximating these sets stem from the fact that their separation or intersection often imply safety or the lack of it.

While numerous algorithms and toolboxes exist for approximating the reach sets, the computational approaches differ depending on the \emph{representation of the set approximants}. In other words, different approaches have different interpretations of what does it mean to approximate a set. For example, parametric approximants seek for a simple geometric primitive (e.g., ellipsoid \cite{kurzhanski1997ellipsoidal,kurzhanskiy2006ellipsoidal,halder2018parameterized,halder2020smallest,haddad2021anytime}, zonotope \cite{girard2005reachability,althoff2011zonotope,althoff2015introduction} etc.) to serve as a proxy for the set. On the other hand, level set approximants \cite{mitchell2000level,mitchell2008flexible} seek for approximating a (value) function whose zero sub-level set is the reach set.

More recent works have specifically advocated  the data-driven learning of reach sets by foregoing models and only assuming access to (numerically or experimentally available) data. These works have also proposed geometric \cite{devonport2020estimating,alanwar2021data} and functional primitives \cite{devonport2021data,thorpe2021learning} as representations. To the best of the authors' knowledge, using the support function as data-driven learning representation for reach sets, as proposed here, is new.

\subsubsection*{Contributions} Our specific contributions are twofold.

\noindent(i) We propose learning a compact set by learning its support function from the (possibly noisy) elements of that set available as finite data. We argue that support function as a learning representation for compact sets is computationally beneficial since several set operations of practical interest have exact functional analogues of operations on corresponding support functions. 

\noindent(ii) We present two algorithms (Sec. \ref{sec:Algorithms}) to learn the support function via sublinear regression: convex quadratic programming, and training an input sublinear neural network that involves nonconvex programming. We demonstrate the comparative performance of these algorithms via numerical examples (Sec. \ref{sec:NumericalResults}) on learning reach sets of controlled dynamics subject to set-valued uncertainties from trajectory data that may be available from simulation or experiments.

This work is organized as follows. In Sec. \ref{sec:SptFnPrelim}, we provide the necessary background on support functions of compact sets, and their correspondence with sublinear functions. Sec. \ref{sec:Algorithms} details the proposed sublinear regression framework using two approaches, viz. solving QP, and training an input sublinear neural network. Sec. \ref{sec:NumericalResults} exemplifies the proposed framework using two numerical case studies on data-driven learning of reach sets. Sec. \ref{sec:conclusions} concludes the paper.

\subsubsection*{Notations} We denote the $d$ dimensional unit sphere $\{\bm{x}\in\mathbb{R}^{d} \mid \|\bm{x}\|_2 = 1\}$ as $\mathbb{S}^{d-1}$, and the standard Euclidean inner product as $\langle\cdot,\cdot\rangle$. The convex hull and the closure of a set $\mathcal{X}$ are denoted as ${\rm{conv}}(\mathcal{X})$ and $\overline{\mathcal{X}}$, respectively. The Legendre-Fenchel conjugate of function $f:\mathcal{X}\mapsto\mathbb{R}$ is $f^{*}(\cdot) := \sup_{\bm{x}\in\mathcal{X}}\{\langle\cdot,\bm{x}\rangle - f(\bm{x})\}$. The conjugate $f^{*}$ is convex whether or not $f$ is. The rectifier linear unit (ReLU) function ${\rm{ReLU}}(\cdot) := \max\{\cdot,0\}$ is defined element-wise. For any natural number $n$, the finite set $\llbracket n\rrbracket:=\{1,2,\hdots,n\}$. 

\section{Background}\label{sec:SptFnPrelim}
\subsection{Support Functions}\label{subsec:sptfn}
For a non-empty set $\mathcal{X}\subseteq\mathbb{R}^{d}$, its \emph{support function} $h_{\mathcal{X}}:\mathbb{S}^{d-1}\mapsto\mathbb{R}\cup\{+\infty\}$ is defined as 
\begin{align}
h_{\mathcal{X}}(\bm{y}) := \underset{\bm{x}\in\mathcal{X}}{\sup}\langle\bm{y},\bm{x}\rangle, \quad\bm{y}\in\mathbb{S}^{d-1}.
\label{defSptFn}    
\end{align}
Being pointwise supremum, $h_{\mathcal{X}}(\cdot)$ is a convex function in its argument, and is finite if and only if \cite[Chap. C.2, Prop. 2.1.3]{hiriart2004fundamentals} $\mathcal{X}$ is bounded. From \eqref{defSptFn}, it also follows that the support function remains invariant under closure of convexification, i.e.,
$$h_{\mathcal{X}}(\cdot) = h_{\overline{{\rm{conv}}}(\mathcal{X})}(\cdot).$$
The Legendre-Fenchel conjugate of the support function $h_{\mathcal{X}}(\cdot)$ is the indicator function
\begin{align}
\bm{1}_{\mathcal{X}}(\bm{x}) := \begin{cases}
0 & \text{if} \;\bm{x}\in\mathcal{X},\\
+\infty & \text{otherwise},
\end{cases}
\label{DefIndicatorFnSet}    
\end{align}
see e.g., \cite[Thm. 13.2]{rockafellar1970convex}. This allows thinking $h_{\mathcal{X}}(\cdot)$ as a functional proxy for the set $\mathcal{X}$, i.e., the support function uniquely determines a compact set up to closure of convexification.

Furthermore, \eqref{defSptFn} has a clear geometric interpretation: $h_{\mathcal{X}}(\bm{y})$ quantifies the signed distance of the supporting hyperplane of compact $\mathcal{X}$ with outer normal vector $\bm{y}$, measured from the origin. This distance is negative if and only if $\bm{y}\in\mathbb{S}^{d-1}$ points into the open halfspace containing origin.

The convergence of compact sets w.r.t. the topology induced by the (two-sided) Hausdorff metric
\begin{align}
\delta_{\mathrm{H}}(\mathcal{P}, \mathcal{Q}):=\max \left\{\sup _{\boldsymbol{p} \in \mathcal{P}} \inf _{\boldsymbol{q} \in \mathcal{Q}}\|\boldsymbol{p}-\boldsymbol{q}\|_2,\sup _{\boldsymbol{q} \in \mathcal{Q}} \inf _{\boldsymbol{p} \in \mathcal{P}}\|\boldsymbol{p}-\boldsymbol{q}\|_2\right\}
\label{DefHausdorff}
\end{align}
for compact $\mathcal{P},\mathcal{Q}$, is equivalent to the convergence of the corresponding support functions. Specifically, a sequence of compact convex sets $\{\mathcal{X}_i\}_{i\in\mathbb{N}}$ converges to a compact convex set $\mathcal{X}$ (denoted as $\mathcal{X}_{i}\rightarrow\mathcal{X}$) in the Hausdorff topology if and only if $h_{\mathcal{X}_i}(\cdot) \rightarrow h_{\mathcal{X}}(\cdot)$ pointwise. For compact convex $\mathcal{P},\mathcal{Q}$, the Hausdorff metric \eqref{DefHausdorff} can be expressed in terms of the respective support functions:
\begin{align}
\delta_{\mathrm{H}}(\mathcal{P}, \mathcal{Q})=\sup _{\bm{y} \in \mathbb{S}^{d-1}}\left|h_{\mathcal{P}}(\bm{y})-h_{\mathcal{Q}}(\bm{y})\right|.
\label{HausdorffSptFn}
\end{align}

Conveniently, operations on sets can be seen as  operations on corresponding support functions. For instance,\\
(i) $\bm{x}\notin\overline{{\rm{conv}}}\left(\mathcal{X}\right)$ if and only if $\exists\:\bm{y}\in\mathbb{S}^{d-1}$ such that $\langle\bm{y},\bm{x}\rangle > h_{\mathcal{X}}(\bm{y})$,\\
(ii) $\mathcal{X}_{1}\subseteq\mathcal{X}_2$ if and only if $h_{\mathcal{X}_1}(\bm{y})\leq h_{\mathcal{X}_2}(\bm{y})$,\\
(iii) $h_{\mathcal{X}_1 + \hdots + \mathcal{X}_r}(\bm{y}) = h_{\mathcal{X}_1}(\bm{y}) + \hdots + h_{\mathcal{X}_r}(\bm{y})$,\\
(iv) $h_{\cup_{i=1}^{r}\mathcal{X}_{i}}(\bm{y}) = \max\{h_{\mathcal{X}_1}(\bm{y}), \hdots, h_{\mathcal{X}_r}(\bm{y})\}$,\\
(v) $h_{\cap_{i=1}^{r}\mathcal{X}_{i}}(\bm{y}) = \underset{{\boldsymbol{y}_1+\ldots+\boldsymbol{y}_r=\boldsymbol{y}}}{\inf} \left\{h_{\mathcal{X}_1}\left(\boldsymbol{y}_1\right)+\ldots+h_{\mathcal{X}_r}\left(\boldsymbol{y}_r\right)\right\}$,\\
(vi) $h_{\bm{A}\mathcal{X}+\bm{b}}(\bm{y}) = \langle\bm{y},\bm{b}\rangle + h_{\mathcal{X}}\left(\bm{A}^{\top}\bm{y}\right)$, $\bm{A}\in\mathbb{R}^{d\times d},\bm{b}\in\mathbb{R}^{d}$.
These correspondence motivate the possibility of using the support functions as computational learning representations for estimating compact sets from data.

We next point out that the support functions have additional structural properties beyond convexity which will be important in our algorithmic development.

\subsection{Sublinear Functions}\label{subsec:sublinfn}
A function $\psi:\mathbb{R}^{d}\mapsto\mathbb{R}\cup\{+\infty\}$ is called \emph{sublinear} if it is convex and positive homogeneous of degree one. The latter condition means that $$\psi(a\bm{x})=a\psi(\bm{x})\quad\forall a>0.$$
Alternatively, $\psi:\mathbb{R}^{d}\mapsto\mathbb{R}\cup\{+\infty\}$ is sublinear if and only if its epigraph is a nonempty convex cone in $\mathbb{R}^{d}\times\mathbb{R}$, see e.g., \cite[Chap. C.1, Prop. 1.1.3]{hiriart2004fundamentals}.

From \eqref{defSptFn}, the support function $h_{\mathcal{X}}(\cdot)$ is both convex and positive homogeneous of degree one, and hence a sublinear function. Conversely, any sublinear function can be viewed as support function of a compact set. The converse follows from the fact \cite[Thm. 8.13]{rockafellar2009variational} that a positive homogeneous convex function can be expressed as pointwise supremum of linear function. Thus, to learn a compact set is to learn its support function, and learning a support function from data leads to \emph{sublinear regression} as opposed to the well-known \emph{convex regression} \cite{holloway1979estimation,lim2012consistency}.

\begin{figure}[t]
\centerline{\includegraphics[width = 0.9\linewidth]{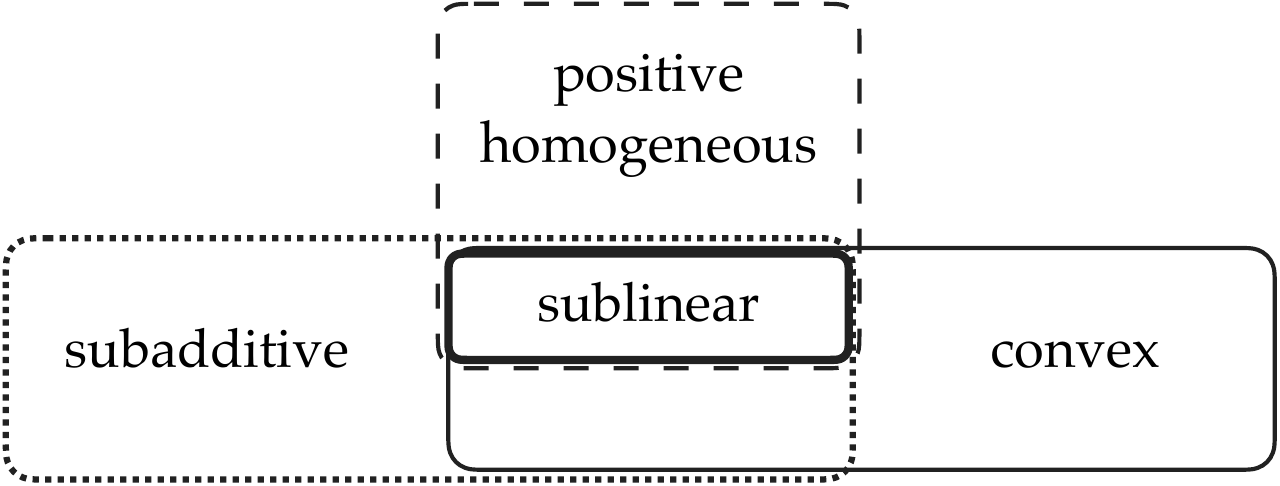}}
\caption{{\small{Euler diagram showing the relationships among the class of convex, positive homogeneous, subadditive and sublinear functions.}}}
\label{figVennDiagram}
\vspace*{-0.25in}
\end{figure}

Any sublinear function (and hence the support function) must also be subadditive, i.e.,
\begin{align}
\psi(\bm{y}+\bm{z}) \leq \psi(\bm{y}) + \psi(\bm{z}) \quad\forall\bm{y},\bm{z}\in\mathbb{R}^{d}.
\label{subadditive}
\end{align}
This can be seen as a joint consequence of convexity and positive homogeneity because specializing convexity to midpoint convexity implies
$$\psi\left(\frac{1}{2}\bm{y}+\frac{1}{2}\bm{z}\right)\leq \frac{1}{2}\psi(\bm{y}) + \frac{1}{2}\psi(\bm{z}),$$
which upon using positive homogeneity yields \eqref{subadditive}. 

Positive homogeneity and subadditivity together is equivalent to sublinearity, which also brings convexity. Following \cite[Chap. C.1]{hiriart2004fundamentals}, an Euler diagram among these classes of functions is shown in Fig. \ref{figVennDiagram}.


\section{Learning Support function via Sublinear Regression}\label{sec:Algorithms}
Following the background in Sec. \ref{sec:SptFnPrelim}, learning the support function of a compact set from data results in a sublinear regression problem, i.e., a regression problem where the to-be-learnt function is constrained to be sublinear. To this end, we next detail the data generation procedure followed by two proposed algorithms for the same.

\subsection{Data Generation}\label{subsec:DataGen}
In our context, the available data comprises of noisy elements of a compact set $\mathcal{X}\subset\mathbb{R}^{d}$. Our data is given by $\{\widehat{\bm{x}}_{j}\}_{j=1}^{n_x} = \{\bm{x}_{j} + \bm{\nu}_{j}\}_{j=1}^{n_x}$ for \emph{deterministic} $\bm{x}_{j}\in\mathcal{X}$, and the i.i.d. samples $\bm{\nu}_j$ are \emph{random} realizations of some noise vector $\bm{\nu}\in\mathbb{R}^{d}$ with zero mean and finite second moment. 

In Sec. \ref{sec:NumericalResults}, we will focus on reach sets of a controlled dynamical system with set-valued input uncertainties, and the finite set $\{\widehat{\bm{x}}_{j}\}_{j=1}^{n_x}$ will correspond to the states resulting from different feasible input sample paths.


We pose the problem of learning the set $\mathcal{X}$ as learning its support function. We propose to learn the latter by computing the \emph{estimate} $\widehat{h}_{\mathcal{X}}(\bm{y})$, $\bm{y} \in \mathbb{S}^{d-1}$, where
\begin{align}
    \widehat{h}_{\mathcal{X}}(\bm{y}_i) = \underset{\widehat{\bm{x}} \in \{\widehat{\bm{x}}_{j}\}_{j=1}^{n_x}}{\sup} \langle\bm{y}_i,\widehat{\bm{x}}\rangle, \quad \forall i\in\llbracket n_y\rrbracket.
    \label{SupFnoisy}
\end{align} 
We seek a sublinear function that ``well fits" the values \eqref{SupFnoisy}.

\subsection{Regression Algorithm}\label{subsec:LRegression}

\subsubsection{QP}\label{subsubsec:QP_LP}
In this approach, we propose a regression algorithm using standard least squares, i.e., by solving the \emph{infinite dimensional} variational problem: 
\begin{align}
\underset{\{{h}_{n_x}:\mathbb{R}^d \to \mathbb{R}\mid {h}_{n_x}(\cdot)\;\text{is sublinear}\}}{\arg\inf}\sum_{i=1}^{n_y} \left(\widehat{h}_{\mathcal{X}}(\bm{y}_i)-h_{n_x}(\bm{y}_i) \right)^2.
\label{LeastSquare}
\end{align}
The least squares problem enjoys the following consistency guarantee.
\begin{theorem}
The minimizer of \eqref{LeastSquare}, ${h}_{n_x}(\cdot)$, almost surely converges to the true support function ${h}_{\mathcal{X}}(\cdot)$ as $n_y,n_x \to \infty$.
\end{theorem}
\begin{proof}
Notice that for  $i=1,\cdots,n_y$, we have
\begin{align}
   d_{n_x,i}&:=\widehat{h}_{\mathcal{X}}(\bm{y}_i)-h_{\overline{{\rm{conv}}}\left(\{{\bm{x}}_{j}\right)\}_{j=1}^{n_x}}(\bm{y}_i) =\bm{y}_i^{\top}\bm{\nu}.
   \label{LSproof1}
\end{align}
Therefore, $d_{n_x,i}$ is i.i.d. with {\small{$\underset{\bm{x},\bm{\nu}}{\mathbb{E}}\big[d_{n_x,i}\big]=0$}} and {\small{$\underset{\bm{x},\bm{\nu}}{\mathbb{E}}\big[d_{n_x,i}^2\big]=1 $}}. As $n_x \to \infty$, we have $h_{\overline{{\rm{conv}}}\left(\{{\bm{x}}_{j}\right)\}_{j=1}^{n_x}}(\cdot) \to h_{{\mathcal{X}}}(\cdot)$ and the least squares problem \eqref{LeastSquare} becomes identical to the one investigated in \cite{lim2012consistency}, where the proof of consistency for $n_y \to \infty$ is provided.
\end{proof}
Given $\{\bm{y}_{i}\}_{i=1}^{n_y}$, solving \eqref{LeastSquare} reduces to solving a \emph{finite dimensional} convex QP \cite[Ch. 6.5.5]{boyd2004convex} as follows. Defining decision variables  $\bm{g}_1,\hdots,\bm{g}_{n_{y}}\in\mathbb{R}^{d}$ (the subgradients) and $\bm{h}:=(h_1,\hdots,h_{n_y})\in \mathbb{R}^{n_y}$, we solve the QP
\begin{align}
&\underset{\bm{g}_1,\hdots,\bm{g}_{n_y} \in \mathbb{R}^{d},\bm{h}\in\mathbb{R}^{n_y}}{\arg\min}  \quad\sum_{i=1}^{n_y}\left(\widehat{h}_{\mathcal{X}}(\bm{y}_i)-h_i \right)^2 \nonumber \\
& \text{subject to}\; h_j \geq h_i+\langle\bm{g}_i,\bm{y}_j-\bm{y}_i\rangle~\forall (i,j) \in \llbracket n_y\rrbracket\times\llbracket n_y\rrbracket.
\label{QP}
\end{align}
We then use the minimizing subgradients 
$\bm{g}_1^{\rm{opt}},\hdots,\bm{g}_{n_{y}}^{\rm{opt}}$ from \eqref{QP} to obtain a piecewise linear (PWL) estimate
\begin{align}
h^{\rm{PWL}}(\cdot)=\underset{i=1,\cdots,n_y}{\max} \bigg\{\widehat{h}_{\mathcal{X}}(\bm{y}_i)+\langle\bm{g}_i^{\rm{opt}},\cdot-\bm{y}
_i\rangle\bigg\}.
\label{PWL}
\end{align}
The rationale behind the PWL construction \eqref{PWL} is as follows. Being both convex and positive homogeneous of degree one, $h^{\rm{PWL}}(\cdot)$ in \eqref{PWL} must be a sublinear function. Recall also the fact we mentioned in Sec. \ref{subsec:sublinfn}: every sublinear function is expressible as pointwise supremum of linear functions \cite[Thm. 8.13]{rockafellar2009variational}.

\subsubsection{Input Sublinear Neural Network}\label{subsubsec:ICNN}
We propose input sublinear neural network (ISNN) as an alternative tool for performing sublinear regression to learn the support function of the set $\mathcal{X} \in \mathbb{R}^d$ in the noisy setting described in Sec. \ref{subsec:DataGen}.

\begin{definition}\label{def:ISNN}
We say a neural network (NN) is input sublinear neural network (ISNN) if the network's output is a sublinear function of the network's input vector $\bm{y}\in\mathbb{R}^{d}$.
\end{definition}

The ISNN structure we propose here is a particular instance of the input convex neural network (ICNN) architecture proposed by Amos \cite{amos2017input} which is suitable for convex regression with guarantees \cite{chen2018optimal}. To ease the exposition, we start with a brief recap of the ICNN architecture.  

An ICNN with $\ell$ layers is designed such that the network output is convex w.r.t. the input vector $\bm{y}\in\mathbb{R}^{d}$. 
Let the width of these layers be $\{n_1,\hdots,n_{\ell}\}$, and let $n_{\ell}:=1$, $n_0 := d$. For all $k\in\llbracket\ell\rrbracket$, the $k$\textsuperscript{th} layer of the network with width $n_k$ has associated weight matrices 
$\bm{W}^{(z)}_{k} \in \mathbb{R}^{n_{k}\times n_{k-1}}_{\geq 0}$, $\bm{W}^{(y)}_{k} \in \mathbb{R}^{n_{k}\times d}$ and bias vector $\bm{b}_{k}\in\mathbb{R}^{n_{k}}$. Furthermore, let $\bm{W}^{(z)}_{1}:=\bm{0}\in\mathbb{R}^{n_1 \times d}$ (zero matrix). Then, for a given input vector $\bm{y} \in \mathbb{R}^d$, the computation for each layer of ICNN involves:
\begin{align}
\bm{z}_1&=\bm{\sigma}\left( \bm{W}^{(y)}_{1}\bm{y}+\bm{b}_1 \right),\nonumber\\
\bm{z}_{k+1}&=\bm{\sigma}\left(\bm{W}^{(z)}_{k+1}\bm{z}_{k}+\bm{W}^{(y)}_{k+1}\bm{y}+\bm{b}_{k+1}\right),\quad k\in\llbracket\ell-2\rrbracket,\nonumber\\
\bm{z}_{\ell} &= \bm{W}^{(z)}_{\ell}\bm{z}_{\ell-1}+\bm{W}^{(y)}_{\ell}\bm{y}+\bm{b}_{\ell},
\label{ICNN}    
\end{align}
where the vector mapping $\bm{\sigma}$ comprises of element-wise application of the same activation function $\sigma(\cdot)$ that is assumed to be convex and non-decreasing. The ICNN model parameters $\bm{W}^{(z)}_{2:\ell}$, $\bm{W}^{(y)}_{1:\ell} $ and $\bm{b}_{1:\ell}$ are determined via the training of the network.

That the ICNN output $\bm{z}_{\ell}$ is guaranteed to be a convex function of the input $\bm{y}$ follows from two facts: \emph{first}, the summation of convex functions is convex; \emph{second}, a function composition where the inner function is convex and the outer function is convex non-decreasing, remains convex.

The main distinction between traditional neural networks and ICNN is that the activation function $\sigma(\cdot)$ in ICNN must be convex and non-decreasing such as ReLU, which is anyway a popular choice for many NN implementations. Furthermore, the weight matrices $\bm{W}^{(z)}_{1:\ell}$ connecting the \emph{feedforward} layers should be elementwise non-negative. The restriction on non-negtaive feedforward weights is compensated with additional \emph{passthrough} links that connect the input layer to each hidden layer \cite{he2016deep} via weight matrices $\bm{W}^{(y)}_{1:\ell}$ comprising of any real values. The bias vectors $\bm{b}_{1:\ell}$ comprise of real entries.


\begin{theorem}\label{Thm:ISNN}
The neural network \eqref{ICNN} is an ISNN, i.e., outputs a sublinear function of a given input vector $\bm{y} \in \mathbb{R}^d$, if $\bm{b}_{1:\ell}=\bm{0}$, and the activation function $\sigma(\cdot)$ is convex, non-decreasing and positive homogeneous of degree one.
\end{theorem}
\begin{proof}
 The computation for each layer of ISNN follows
\begin{align}
\bm{z}_1&=\bm{\sigma}\left( \bm{W}^{(y)}_{1}\bm{y}\right),\nonumber\\
\bm{z}_{k+1}&=\bm{\sigma}\left(\bm{W}^{(z)}_{k+1}\bm{z}_{k}+\bm{W}^{(y)}_{k+1}\bm{y}\right),\quad k\in\llbracket\ell-2\rrbracket,\nonumber\\
\bm{z}_{\ell} &= \bm{W}^{(z)}_{\ell}\bm{z}_{\ell-1}+\bm{W}^{(y)}_{\ell}\bm{y},
\label{ISNN}    
\end{align}
with elementwise non-negative matrices $\bm{W}^{(z)}_{1:\ell}$. Each layer performs composition of non-negative sums of linear functions followed by a convex, non-decreasing map. Thanks to the positive homogeneity of $\sigma(\cdot)$, this structure preserves the sublinearity w.r.t. the input vector $\bm{y}\in\mathbb{R}^{d}$.
\end{proof}

Examples of activation functions $\sigma(\cdot)$ which satisfy the conditions in Thm. \ref{Thm:ISNN} include ReLU, leaky ReLU and parametric ReLU with a positive parameter. In the numerical simulations reported here, we use the ReLU activation. 

For all ISNN implementations in the following Sec. \ref{sec:NumericalResults}, we use 5 hidden layers with the respective number of neurons $(5,20,50,20,5)$, and the mean squared error as the loss function. For the nonconvex training of the ISNN, we use the Adam optimizer \cite{kingma2014adam}, and project the Adam updates to the nonnegative orthant as
$$\bm{W}^{(z)}_{k}\mapsto {\rm{ReLU}}\left(\bm{W}^{(z)}_{k}\right) \quad \forall k\in\llbracket\ell\rrbracket.$$

\begin{figure}[tb]
    \centering
\includegraphics[width=\linewidth]{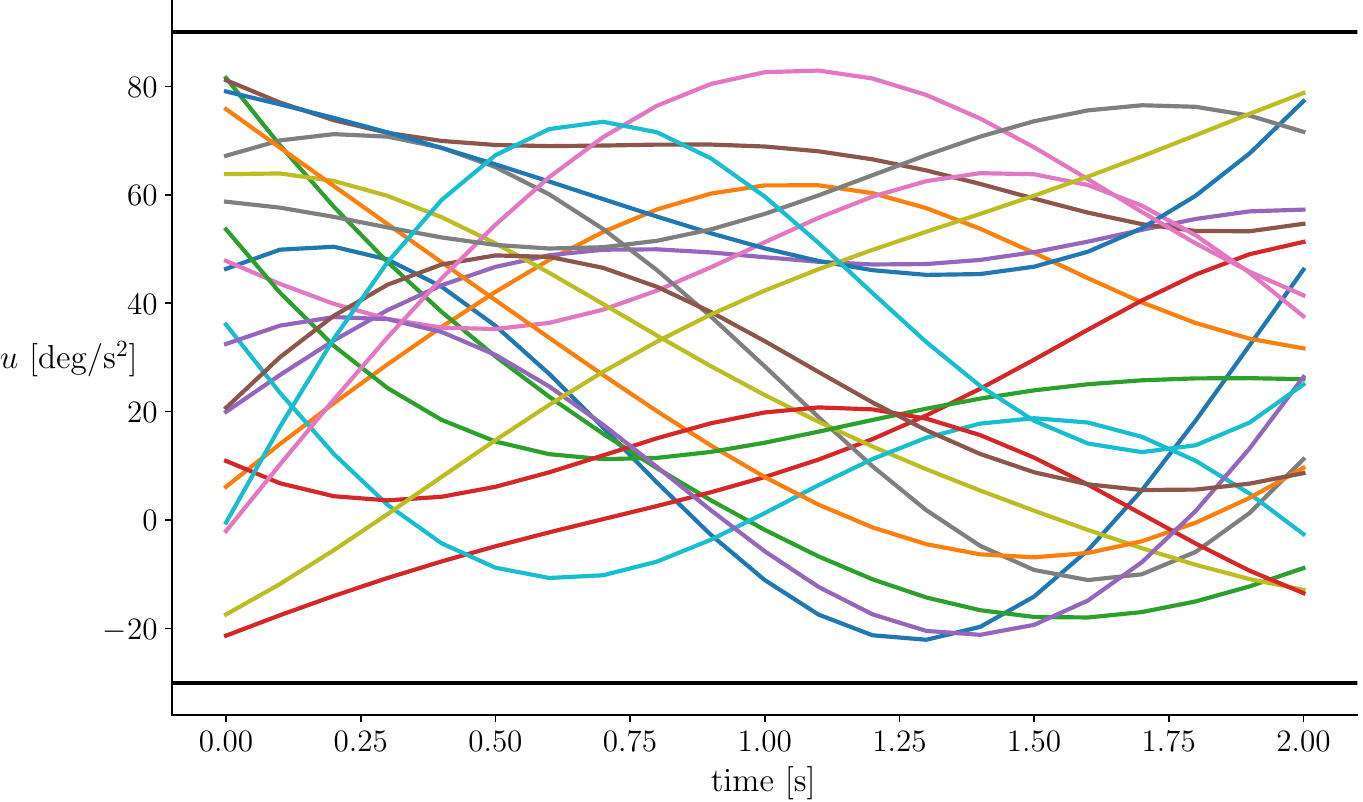}
    \caption{{\small{A subset of the 500 input sample paths $u_i(t) \in \mathcal{U}:=[-30^{\circ}, 90^{\circ}]   ~\text{deg}/\text{s}^2$ obtained from the constrained GP, that are used in Sec. \ref{subsec:Example1}. The dark horizontal lines denote the input constraints.}}}
\vspace*{-0.2in}
\label{fig:InputSamplePaths}
\end{figure}
\begin{figure*}[tb]
    \centering
    \includegraphics[width=\textwidth]{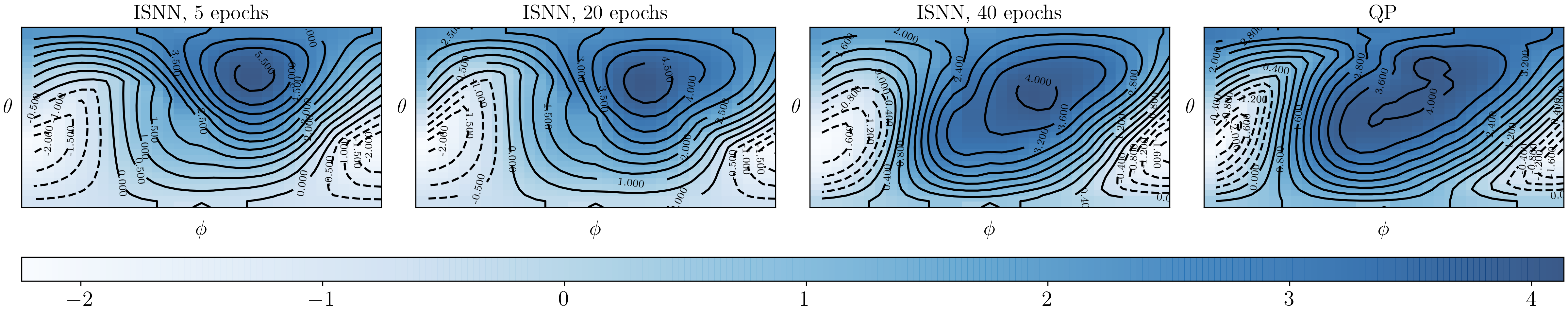}
    \caption{{\small{Contour plots of the estimated support function of the reach set $\mathcal{X}_t\subset {R}^{2}\times \mathbb{S}^{1}$ at $t=2$ s for the example in Sec. \ref{subsec:Example1}, resulting from the proposed sublinear regression. All subfigures are plotted over the spherical coordinates $(\phi,\theta) \in (-\pi,\pi]\times [-\pi/2,\pi/2]$. Different epochs for the ISNN results are indicated in the respective subfigures.}}}
\vspace*{-0.2in}
\label{Ex3d}
\end{figure*}

\section{Numerical Results}\label{sec:NumericalResults}
In this Section, we illustrate the proposed sublinear regression methods to learn the reach sets for two example controlled nonlinear dynamics with input uncertainties. We clarify here that we only use the controlled ODE models for data generation purpose, i.e., our computation is data-driven and is agnostic to the structural specificities of the models. All computation were performed in a MacBook Pro with 2.6 GHz 6-Core Intel i7 processor with 16 GB of memory.


\subsection{Sampling}\label{subsec:sampling}

We now describe the sampling procedure for both the numerical examples that follows.

For training data, we use the constrained Gaussian process (GP) to generate $n_x=500$ sample paths $\{\bm{u}_{i}(t)\}_{i=1}^{n_x}$ where each $\bm{u}_{i}(t)\in \mathcal{U}$. For instance, when $\mathcal{U}\subset\mathbb{R}^{m}$ is a hyperrectangle, then to ensure that the range of the sampled functions are in $\mathcal{U}$, we use the truncated multivariate Gaussian distributions \cite[Sec. 2.2]{robert1995simulation} for generating the constrained GP sample paths via the Gibbs Metropolis-Hastings Markov Chain Monte Carlo sampler. For hyperrectangle $\mathcal{U}$, we set the mean of the GP to be the center of the hyperrectangle, and the covariance function $k(\bm{x},\tilde{\bm{x}}):=\exp\left(-\|\bm{x}-\tilde{\bm{x}}\|_2^2/(2\ell^2)\right)$ with $\ell=0.7$. 

We then generate $n_y=200$ uniformly random unit vectors $\{\bm{y}_i\}_{i=1}^{n_y}$ in $\mathbb{R}^{d}$, and use \eqref{SupFnoisy} to form the training data
$\{(\bm{y}_i,\widehat{h}_{\mathcal{X}}(\bm{y}_i))\}_{i=1}^{n_y}$.

\subsection{Dubin's Car}\label{subsec:Example1}
We consider the controlled dynamics for the Dubin's car 
\begin{equation}
\dot{x}_1 =v\cos x_3,
\quad \dot{x}_2 =v\sin x_3,
\quad \dot{x}_3 =u,
\label{ex1}
\end{equation}
with the bounded (angular velocity) input $u(t) \in \mathcal{U}:=[-30^{\circ}, 90^{\circ}]   ~\text{deg}/\text{s}^2$ for all $t\geq 0$, and constant translational velocity $v=2 ~\text{m}/\text{s}$. The state vector $\left(x_1,x_2,x_3\right)^{\top}\in\mathbb{R}^{2}\times \mathbb{S}^{1}$ comprises of the longitudinal position, the lateral position, and the heading angle, respectively. We suppose that the initial set $\mathcal{X}_0 = \{\bm{0}\}$ is singleton, i.e., zero initial condition with no uncertainties.

Fig. \ref{fig:InputSamplePaths} shows a subset of the constrained GP-generated input sample paths $\{u_{i}(t)\}_{i=1}^{n_x}\in\mathcal{U}$ for this example.
\vspace{- 0.2 in}
\begin{center}
\begin{table}
\centering
\begin{tabular}{c c c c c} 
 \hline
 Instance &\quad \quad & ISNN, 30 epochs  & \quad \quad & QP\\
 \hline\hline
 1& \quad & 6.76 &\quad& 60.78\\ 
 \hline
 2& \quad & 6.66 &\quad& 60.52\\
 \hline
 3 & \quad& 6.88 &\quad& 63.54 \\
 \hline
 4 & \quad& 6.68 &\quad& 67.02\\
 \hline
  5& \quad & 6.63 &\quad& 66.55\\
 \hline
  6& \quad & 7.09 &\quad& 66.92\\
 \hline
  7& \quad & 6.63 &\quad& 73.82\\
 \hline
  8& \quad & 6.66 &\quad& 74.56\\
 \hline
 9& \quad & 6.66 &\quad& 71.91 \\
 \hline
 10& \quad & 6.98 &\quad& 69.60 \\
 \hline
\end{tabular}
\caption{{\small{Computational times [s] incurred by the sublinear regression for the example in Sec. \ref{subsec:Example1} for 10 different random sampling instances.}}} 
\label{tab:table}
\end{table}
\end{center}

Fig. \ref{Ex3d} depicts the learnt support functions over $\mathbb{S}^2$, visualized in spherical coordinates, for the reach set $\mathcal{X}_t$ of dynamics \eqref{ex1} at $t = 2$ s, using the proposed sublinear regressions ISNN and QP. The support functions resulting from the ISNN are depicted for three different number of epochs: 5, 20 and 40 epochs. As expected, increasing the number of epochs makes the ISNN solutions approach the QP solution.

Overall, we find that solving the QP results in a more robust estimation w.r.t. the input data compared to the ISNN. However, since the  number of constraints in QP \eqref{QP} is quadratic in $n_y$, we observe that the training time for QP is considerably higher (approx. 10 times) than that of ISNN. A computational time comparison is reported in Table \ref{tab:table} for 10 different random sampling instances with fixed cardinality $(n_x,n_y)=(500,200)$ while keeping all other simulation settings fixed.

\begin{figure}[tb]
    \centering
\includegraphics[width=\linewidth]{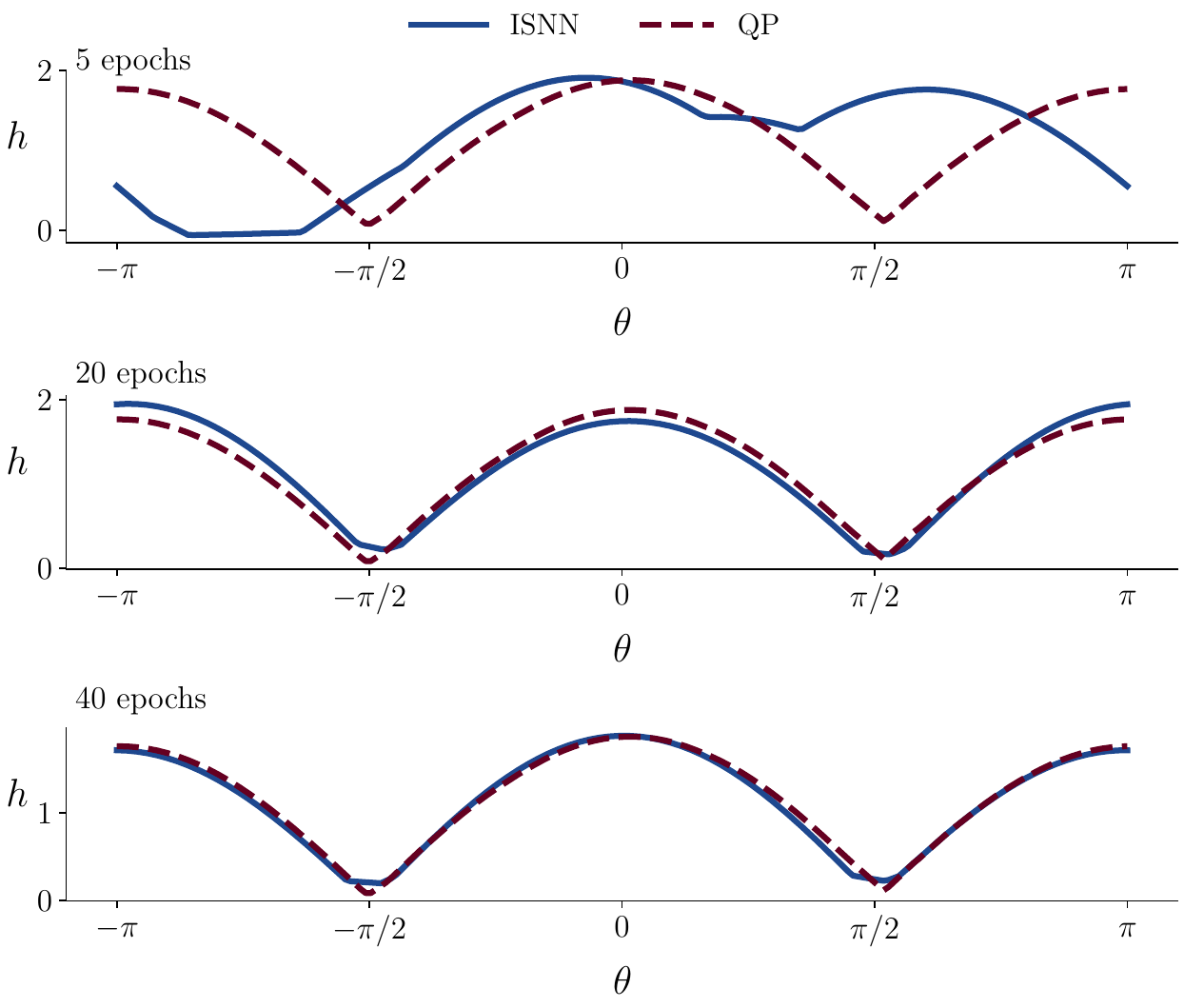}
    \caption{{\small{The estimated support function of the projection of the reach set $\mathcal{X}_t$ onto the position coordinates: ${\rm{proj}}\left(\mathcal{X}_t\right)\subseteq\mathbb{R}^2$ at $t=2$ s for the example in Sec. \ref{subsec:Example2}, resulting from the proposed sublinear regression. All subfigures are plotted over the polar coordinate $\theta \in (-\pi,\pi]$. Different epochs for the ISNN results are indicated in the respective subfigures.}}}
\vspace*{-0.2in}
\label{Ex2}
\end{figure}
\begin{figure}[tb]
    \centering
\includegraphics[width=\linewidth]{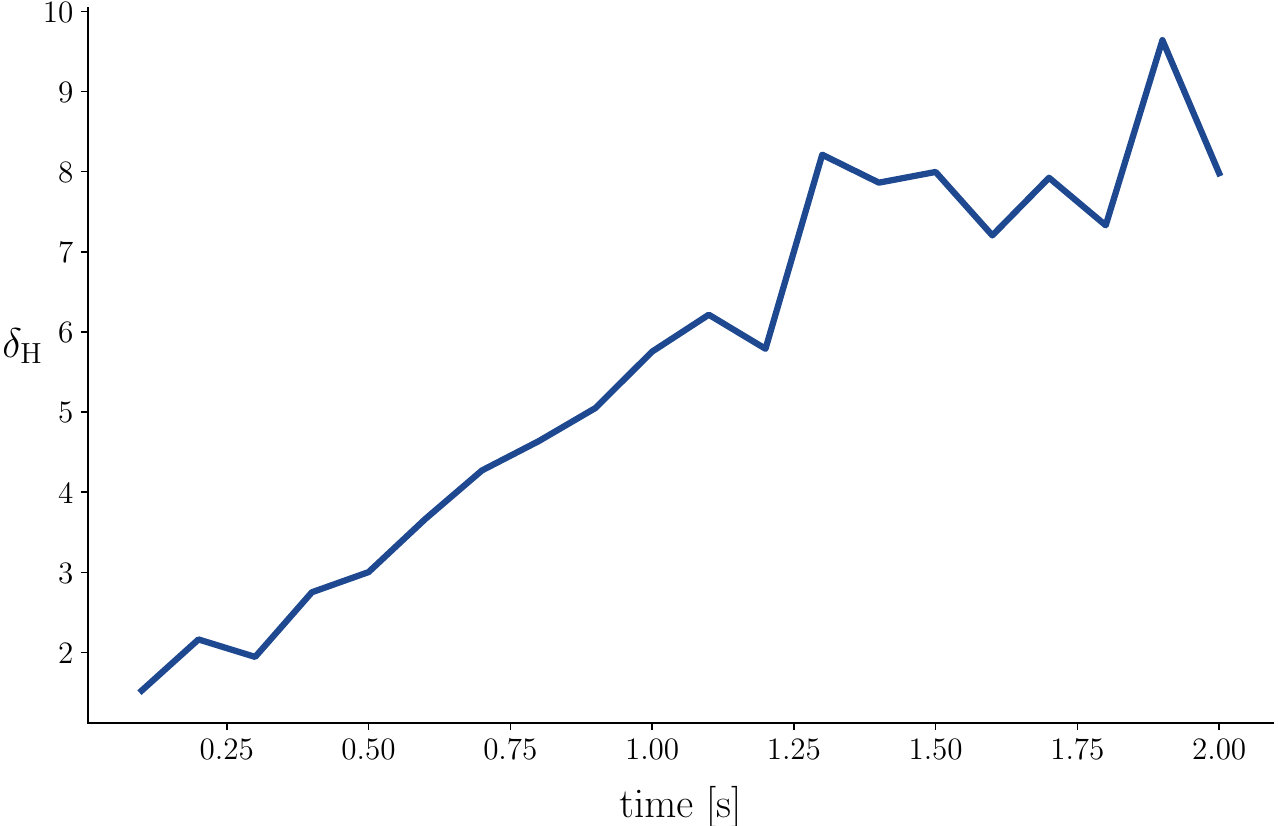}
    \caption{{\small{The Hausdorff distance \eqref{HausdorffSptFnEx} between ${\rm{proj}}\left(\mathcal{X}^{\texttt{A}}_{\tau}\right),{\rm{proj}}\left(\mathcal{X}^{\texttt{B}}_{\tau}\right) \in \mathbb{R}^2$ for $0 \leq \tau \leq 2$ s for the example in Sec. \ref{subsec:Example2}. The corresponding support functions $\widehat{h}^{\texttt{A}}_{{\rm{proj}}\left(\mathcal{X}_\tau\right)}$ and $\widehat{h}^{\texttt{B}}_{{\rm{proj}}\left(\mathcal{X}_\tau\right)}$ are learnt by performing the sublinear regressions using ISNN with 40 epochs for both agents.}}}
\vspace*{-0.2in}
\label{fig:HausdorffDist}
\end{figure}

\subsection{Kinematic Bicycle}\label{subsec:Example2}
We next consider the controlled dynamics for the kinematic bicycle \cite[p. 22]{rajamani2011vehicle}
\begin{align}
\begin{split}
\dot{x}_1&=x_3\cos(x_4+\beta),\quad \dot{x}_2=x_3\sin(x_4+\beta),\\
\dot{x}_3&=u_1,\quad \quad \quad \quad \quad \quad \!~ \dot{x}_4=x_3\sin(\beta)/1.5,\\
\end{split}
\label{ex2}
\end{align}
wherein the sideslip angle
$$\beta :=\arctan\left( 0.6 \tan u_2\right).$$
The state vector $\left(x_1,x_2,x_3,x_4\right)^{\top}\in\mathbb{R}^{3}\times \mathbb{S}^{1}$ comprises of the 2D inertial position
$(x_1,x_2)$ denoting the vehicle’s center of mass, its speed $x_3$, and
the vehicle’s inertial heading angle $x_4$. The control vector 
comprises of the acceleration $u_1$, and the front steering wheel
angle $u_2$. We consider $(u_1,u_2) \in \mathcal{U} := [-1,1] ~\text{m}/\text{s}^2\times [-10^{\circ},10^{\circ}]$. We fix $\mathcal{X}_0 = \{\bm{0}\}$ (singleton).

We follow the sampling procedure detailed in Sec. \ref{subsec:sampling} with $(n_x,n_y)=(500,200)$. We estimate the support function of the projection of the reach set $\mathcal{X}_t\subset\mathbb{R}^3\times\mathbb{S}^1$ for \eqref{ex2} onto the position coordinates, i.e., ${\rm{proj}}\left(\mathcal{X}_t\right)\subseteq\mathbb{R}^2$ at $t=2$ s. 

As in Sec. \ref{subsec:Example1}, for sublinear regression, we use the  QP, and the ISNN with three different number of epochs: 5, 20 and 40 epochs. Consistent with the observation made before, Fig. \ref{Ex2} reveals that as the number of epochs for ISNN increases, the ISNN estimates approach the QP estimate.

To further illustrate the use of support function learning representations for the reach sets, consider the reach sets of two agents $\texttt{A}$ and $\texttt{B}$ with identical dynamics \eqref{ex2}, respective inputs $(u_1^\texttt{A},u_2^\texttt{A})\in\mathcal{U}^\texttt{A}$ and $(u_1^\texttt{B},u_2^\texttt{B})\in\mathcal{U}^\texttt{B}$, and singleton initial conditions $\{\bm{x}_0^{\texttt{A}}\}, \{\bm{x}_0^{\texttt{B}}\}$ with   
 \begin{subequations}
 \begin{align}
 &\mathcal{U}^\texttt{A} := [-1,1] ~\text{m}/\text{s}^2\times [-10^{\circ},10^{\circ}],\\
  &\mathcal{U}^\texttt{B} := [-1.2,1] ~\text{m}/\text{s}^2\times [-2^{\circ},15^{\circ}],\\
  &\bm{x}_0^{\texttt{A}}:=(-1,1,10,0.1)^{\top},~\bm{x}_0^{\texttt{B}}:=(0,0,8,-0.5)^{\top}.
\end{align} 
 \end{subequations}
For $0\leq \tau \leq t$, let us denote the respective reach sets as $\mathcal{X}^{\texttt{A}}_{\tau}, \mathcal{X}^{\texttt{B}}_{\tau}\subset\mathbb{R}^{3}\times\mathbb{S}^1$.

We wish to estimate the Hausdorff distance $\delta_{\mathrm{H}}(\tau)$, $0\leq \tau \leq t$, between the projections of the reach sets of agents $\texttt{A}$ and $\texttt{B}$ on the $(x_1,x_2)$ inertial position plane. From \eqref{HausdorffSptFn},

\begin{align}
\delta_{\mathrm{H}}(\tau)=\sup _{\bm{y} \in \mathbb{S}^{d-1}}\left|\widehat{h}_{{\rm{proj}}\left(\mathcal{X}_\tau^{\texttt{A}}\right)}(\bm{y})-\widehat{h}_{{\rm{proj}}\left(\mathcal{X}_\tau^{\texttt{B}}\right)}(\bm{y})\right|, \quad 0\leq \tau \leq t.
\label{HausdorffSptFnEx}
\end{align}
Fig. \ref{fig:HausdorffDist} shows the evolution of the Hausdorff distance between the sets ${\rm{proj}}\left(\mathcal{X}_\tau^{\texttt{A}}\right)$ and ${\rm{proj}}\left(\mathcal{X}_\tau^{\texttt{B}}\right)$ for $0\leq \tau\leq 2$ s. For this computation, We performed the sublinear regression using ISNN with 40 epochs for the noisy measurements of support functions $\widehat{h}^{\texttt{A}}_{{\rm{proj}}\left(\mathcal{X}_\tau\right)}(\bm{y})$ and $\widehat{h}^{\texttt{B}}_{{\rm{proj}}\left(\mathcal{X}_\tau\right)}(\bm{y})$ generated following the steps in Sec. \ref{subsec:sampling}. Both finite sampling and numerical approximation errors in ISNN regression  contribute to the fluctuations observed in Fig. \ref{fig:HausdorffDist}.


\section{Conclusions}\label{sec:conclusions}
In this work, we propose data-driven learning of compact sets in general, and reach sets in particular, by learning the corresponding support function representations. We point out an equivalence between the support functions and the class of sublinear functions, and propose leveraging the same for performing sublinear regression. We numerically demonstrate and compare two approaches: the \emph{first} involves convex quadratic programming (QP), and the \emph{second} being ISNN that involves nonconvex programming. Our numerical experiments reveal that among the two, the ISNN is numerically faster but the QP solution is more robust and comes with consistency guarantee. We empirically observe that ISNN with modest number of epochs can be a practical alternative to QP without incurring as much computational cost as the latter.




\bibliographystyle{IEEEtran}
\bibliography{References.bib}

\end{document}